\documentclass[pra,aps, 10pt, twocolumn,showpacs,superscriptaddress,nofootinbib]{revtex4-2}
\usepackage[table, svgnames, dvipsnames]{xcolor}
\usepackage{amsmath, amsfonts, amssymb}
\usepackage{bbm}
\usepackage{mathtools}
\usepackage{bm}
\usepackage[normalem]{ulem}
\usepackage{braket}
\usepackage{hyperref}
\usepackage{physics}
\newcommand\numberthis{\addtocounter{equation}{1}\tag{\theequation}}
\usepackage[]{lineno} 
\usepackage{stmaryrd}
\usepackage{mathbbol}
\usepackage[sans]{dsfont}
\usepackage{float}

\allowdisplaybreaks

\usepackage{graphicx}
\usepackage{tensor}
\usepackage{mathrsfs}

\usepackage{amsthm}
\theoremstyle{plain}
\newtheorem*{theorem*}{Theorem}
\newtheorem{theorem}{Theorem}
\newtheorem{lemma}[theorem]{Lemma}

\theoremstyle{definition}
\newtheorem{definition}{Definition}  
\newtheorem*{definition*}{Definition}        

\usepackage{array}
\usepackage{tabularx}

\usepackage{lettrine}
\usepackage{setspace}
\usepackage[T1]{fontenc}

\usepackage{booktabs}

\definecolor{blue}{HTML}{03045E}
\definecolor{green}{HTML}{6a994e}
\definecolor{lightblue}{HTML}{134074}

\usepackage{orcidlink}
\hypersetup{
	colorlinks=true,
	citecolor= blue,
	linkcolor=blue,
	filecolor=blue,      
	urlcolor=blue,
	pdftitle={Overleaf Example},
	pdfpagemode=FullScreen,
}

\usepackage[
    top=3cm,    
    bottom=3cm, 
    left=2cm, 
    right=2cm 
]{geometry}

\usepackage{titlesec}
\titleformat{\section}
{\sffamily\bfseries} 
{\thesection} 
{1em} 
{} 

\titleformat{\subsection}
{\sffamily\bfseries} 
{\thesection \thesubsection} 
{1em} 
{} 

\titleformat{\subsubsection}
  {\normalsize\itshape\centering} 
  {\thesubsubsection.} 
  {1em} 
  {} 

\newcommand{\secref}[1]{\hyperref[#1]{Sec.~\ref{#1}}}
\renewcommand{\eqref}[1]{\hyperref[#1]{Eq.\,(\ref{#1})}}
\newcommand{\figref}[1]{\hyperref[#1]{Fig.~\ref{#1}}}
\newcommand{\figaref}[1]{\hyperref[#1]{Fig.~\ref{#1}(a)}}
\newcommand{\figbref}[1]{\hyperref[#1]{Fig.~\ref{#1}(b)}}
\newcommand{\figcref}[1]{\hyperref[#1]{Fig.~\ref{#1}(c)}}
\newcommand{\figdref}[1]{\hyperref[#1]{Fig.~\ref{#1}(d)}}

\usepackage[T1]{fontenc}
\usepackage[utf8]{inputenc}

\titlespacing\section{2pt}{12pt plus 0pt minus 0pt}{5pt plus 0pt minus 0pt}
\titlespacing\subsection{2pt}{5pt plus 0pt minus 0pt}{3pt plus 0pt minus 0pt}
\titlespacing\subsubsection{2pt}{5pt plus 0pt minus 0pt}{3pt plus 0pt minus 0pt}
\usepackage{lipsum} 

\makeatletter
\let\@afterindenttrue\@afterindentfalse
\makeatother

\begin{document}
\title{\textsf{Achieving the Heisenberg limit using fault-tolerant quantum error correction}}

\author{Himanshu Sahu}\email{hsahu@pitp.ca}
\affiliation{Perimeter Institute for Theoretical Physics, Waterloo, Ontario N2L 2Y5, Canada}
\affiliation{Department of Physics and Astronomy and Institute for Quantum Computing, University of Waterloo, Ontario N2L 3G1, Canada}

\author{Qian Xu} 
\affiliation{Institute for Quantum Information and Matter, Caltech, Pasadena, CA 91125, USA}
\affiliation{Walter Burke Institute for Theoretical Physics, Caltech, Pasadena, CA 91125, USA}

\author{Sisi Zhou}\email{sisi.zhou26@gmail.com}
\affiliation{Perimeter Institute for Theoretical Physics, Waterloo, Ontario N2L 2Y5, Canada}
\affiliation{Department of Physics and Astronomy and Institute for Quantum Computing, University of Waterloo, Ontario N2L 3G1, Canada}
\affiliation{Department of Applied Mathematics, University of Waterloo, Ontario N2L 3G1, Canada}

\begin{abstract}
    Quantum effect enables enhanced estimation precision in metrology, with the Heisenberg limit (HL) representing the ultimate limit allowed by quantum mechanics. Although the HL is generally unattainable in the presence of noise, quantum error correction (QEC) can recover the HL in various scenarios. A notable example is estimating a Pauli-$Z$ signal under bit-flip noise using the repetition code, which is both optimal for metrology and robust against noise. However, previous protocols often assume noise affects only the signal accumulation step, while the QEC operations---including state preparation and measurement---are noiseless. To overcome this limitation, we study fault-tolerant quantum metrology where all qubit operations are subject to noise. We focus on estimating a Pauli-$Z$ signal under bit-flip noise, together with state preparation and measurement errors in all QEC operations. We propose a fault-tolerant metrological protocol where a repetition code is prepared via repeated syndrome measurements, followed by a fault-tolerant logical measurement. We demonstrate the existence of an error threshold, below which errors are effectively suppressed and the HL is attained. 

\vspace{2em}
\end{abstract}

\maketitle

\vspace{1cm}

\section{Introduction}\label{sec:introduction}

Quantum metrology~\cite{doi:10.1126/science.1104149,giovannetti2011advances,PhysRevLett.96.010401} deals with the task of parameter estimation using quantum resources such as entanglement to achieve higher statistical precision than purely classical approaches. The process typically involves three stages: the preparation of a probe, allowing it to interact with the signal that encodes the parameter of interest, and finally performing a measurement of the probe. The resulting measurement outcomes are then used to estimate the parameter. Quantum metrology has found wide-ranging applications, including gravitational-wave detection~\cite{PhysRevD.23.1693,PhysRevA.33.4033,PhysRevA.88.041802,ligo_2011,Aasi_2013}, frequency spectroscopy~\cite{PhysRevA.46.R6797}, atomic clocks~\cite{Katori2011,PhysRevA.72.042301,10.1063/1.1797561}, and other high-precision measurements~\cite{doi:10.1126/science.1170730,Mauranyapin2017}.

Quantum mechanics sets ultimate limits in precision through the `Heisenberg limit' (HL), which scales as $n^{-1}$~\cite{PhysRevLett.96.010401,Leibfried2004}, where $n$ denotes the total number of probes used in an experiment. In principle, the HL is attainable, for instance, using the Greenberger-Horne-Zeilinger (GHZ) state in atomic systems or the NOON state in photonic systems. In practice, however, environmental decoherence substantially reduces this quantum enhancement~\cite{Sekatski2017quantummetrology,PhysRevX.7.041009,Zhou2018,Fujiwara_2008,Escher2011,4655455,Demkowicz-Dobrzański2012,PhysRevLett.113.250801,Kołodyński_2013}, restricting precision to $n^{-1/2}$, referred to as the `standard quantum limit' (SQL) or `shot noise' in quantum optics. Over the past few decades, significant progress has been made in identifying strategies to retain quantum enhancement under noise~\cite{Zhou2018,PhysRevX.7.041009}, employing a variety of theoretical and experimental tools.

Quantum error correction (QEC) is a powerful tool for establishing bounds on achievable precision in quantum metrology~\cite{PhysRevLett.112.150802,PhysRevLett.112.150801,PhysRevLett.112.080801,ozeri2013heisenberglimitedmetrologyusing,unden2016quantum,Lu2015,PhysRevA.95.032303,Zhou2018}. QEC codes store the logical information in an entangled state of many physical qubits. Error detection and correction are enabled through syndrome measurements: multi-qubit measurements that extract information about errors without disturbing the encoded quantum information. Depending on the specific code, syndrome measurements can generally reveal whether an error has occurred, as well as its location and type. 

In the task of estimating a signal Hamiltonian under general Markovian noise, it was known that the HL can be achieved if and only if the signal Hamiltonian lies outside the noise span---a condition known as the `Hamiltonian-not-in-Lindblad-span' (HNLS) criterion~\cite{Zhou2018,PhysRevX.7.041009}, or the `Hamiltonian-not-in-Kraus-span' (HNKS) criterion~\cite{PRXQuantum.2.010343}. 
When the condition is satisfied, one can construct a QEC code that suppresses noise without eliminating the signal~\cite{Zhou2018,PRXQuantum.2.010343}. If it is violated, the achievable precision is reduced to the SQL. In both cases, QEC  protocols can achieve the optimal estimation precision~\cite{Zhou_2020,Demkowicz_Dobrza_ski_2014,Wan_2022}. 

Although QEC codes provide optimal strategies in principle, they typically rely on the availability of noiseless ancillary systems and fast, accurate quantum processing. These assumptions present major obstacles to the practical implementation of the QEC-based schemes. 
The requirement of noiseless ancilla can be removed in some physically relevant scenarios, e.g., estimating Pauli-Z signal under bit-flip noise where quantum repetition codes are used~\cite{PhysRevLett.112.080801,PhysRevLett.112.150801, PhysRevLett.112.150802,unden2016quantum}. More general ancilla-free schemes have also been developed recently~\cite{Layden_2019,Peng2020,zhou2024achieving}. 
In these schemes, quantum information is redundantly encoded across multiple qubits, suppressing noise at the logical level and thereby preserving the HL.  

There remains, however, the strong assumption that one can perform error-free quantum operations, such as syndrome extraction or single-qubit measurements. This assumption is crucial since even one faulty probe measurement can reduce the achievable precision to the SQL. In this work, we aim to relax this requirement and ask whether a fault-tolerant regime for quantum metrology can be established, in analogy with fault-tolerant quantum computation~\cite{548464,PhysRevA.57.127}. Concretely, we assume that any \textit{location} in the protocol may fail, including idle (wait) steps. By a location (or \textit{physical location}), we mean any instance of a basic component---namely, a \textit{physical qubit}, a \textit{physical gate}, or a \textit{physical measurement}\footnote[3]{We use \textit{physical} for emphasize operations at the hardware level, in contrast to logical operations acting on encoded qubits.}. Errors occurring at a given location may propagate to others, thereby creating additional \textit{fault locations.} We formalize these assumptions by defining the circuit-level noise model considered in this work as follows.

\begin{definition}[Circuit-level bit-flip noise model]\label{def:noise model}
We consider a circuit-level noise model in which all errors are Pauli-$X$ (bit-flip) errors and occur independently in space and time. Specifically:
\begin{itemize}
    \item \textbf{State Preparation}: An intended initialization of a qubit in $|0\rangle$ (or $|+\rangle $) results in $|1\rangle$ (or $|-\rangle$) with probability $p_\text{prep}$.
    \item \textbf{Measurements}: Each single-qubit measurement is a projective measurement in the computational ($Z$) basis, with the reported eigenvalue flipped with probability $q$. 
    \item \textbf{Gate and idle errors}: After every single-qubit gate, two-qubit gate, and idle (identity) operation, each involved qubit independently suffers a Pauli-$X$ error with probability $p$.
\end{itemize}
    All error events are assumed to be independent and uncorrelated. No other error channels are considered.\footnote[5]{A more general circuit-level noise model is discussed in Appendix\,\ref{appendix: correlated error}, in which each CNOT gate is followed by a two-qubit Pauli-$X$ error drawn from $\{I\otimes X, X\otimes I, X\otimes X\}$ with probability $p_{\mathrm{CNOT}}$.}
\end{definition}

At the phenomenological level, we model this by assuming that each component suffers an error independently with some probability. Note that while phenomenological noise is not generally equivalent to circuit-level noise due to effects such as hook errors, our bit-flip noise model inherently suppresses such correlated errors, rendering the phenomenological and circuit-level descriptions equivalent for the purpose of our protocol. 

Our main result is to show that the HL can be retained in a fault-tolerant setting where we identify a threshold of the physical error rate below which the precision continues to scale as the HL. We consider estimating a Pauli-$Z$ Hamiltonian that evolves for a short time on $n$ probe qubits and the entire process is subject to bit-flip and state preparation and measurement (SPAM) errors. Informally, we state the theorem as:
\begin{theorem}[Error-threshold for Heisenberg-limited sensing]
In the task of estimating a Pauli-$Z$ Hamiltonian with $n$ probe qubits subjected to local, independent SPAM noise and circuit-level bit-flip errors, there exists a nonzero threshold on the physical error rates. Below this threshold, there exists a sensing protocol whose estimation precision achieves Heisenberg-limited scaling.
\end{theorem}

Our protocol can be divided into two main parts: probe preparation and measurement. We use a quantum repetition code to prepare probes in an imperfect GHZ state---which approximates the behavior of the optimal GHZ state. As the syndrome measurements are faulty, we repeat these measurements and perform a global recovery method. The probe state then interacts with the signal, therefore, encodes the parameter into the probe state. To perform probe measurement in a fault-tolerant manner, we repeatedly measure the Pauli-$X$ operator non-destructively, facilitated by ancilla qubits, on each probe qubit to construct the logical $X$ measurement. In both the probe preparation and measurement stage, state preparation errors of physical qubits (probe or ancilla) are suppressed by repeatedly measuring the relevant single-qubit observables.  Moreover, we show that it is sufficient to achieve fault-tolerance with only a number of repeated measurements logarithmic in $n$ in total, thereby, only logarithmic overhead in circuit-depth. Our work provides a first example where the HL can be recovered in a fault-tolerant manner where noisy QEC operations are used and bridges the gap between theoretical optimality of QEC and its practical robustness in quantum metrology experiments.

\section{Preliminaries}\label{sec:preliminaries}

\subsection{Classical estimation theory}\label{subsec:classical-estimation-theory}

Consider $(x_1,x_2,\ldots ,x_N)$ to be independent and identically distributed (i.i.d.) random variables drawn from the probability distribution $\mathscr{P}(x;\theta)$, where $\theta \in \Theta\subseteq\mathbb{R}$, $\Theta$ denotes the parameter space (i.e., the set of all possible values of $\theta$), and $\theta$ is the true but unknown parameter value. An estimator $\hat{\theta}$ is \textit{locally unbiased} at $\theta_0 \in \Theta$ if 
\begin{equation}
    \left.\mathbb{E}_\theta[\hat{\theta}]\right|_{\theta_0} = \theta_0,\quad \frac{\partial}{\partial \theta} \left. \mathbb{E}_\theta[\hat{\theta}]\right|_{\theta_0} = 1\,.
\end{equation}
where $\mathbb{E}_\theta[\cdot ]$ denotes the expectation value with respect to the probability distribution $\mathscr{P}(x;\theta)$. The Cram\'{e}r-Rao bound (CRB)\,\cite{Rao1992, Cramer1945} provides, for any locally unbiased estimator, a lower bound on the estimator error $\Delta \hat{\theta}$, given by 
\begin{equation}\label{eq:CRLB}
    (\Delta \hat{\theta})^2 := \mathbb{V}_\theta[\hat{\theta}] \geq \frac{1}{N\mathcal{F}_\text{cl}(\theta)},
\end{equation}
where $\mathbb{V}_\theta[\cdot ]$ denotes the variance over the probability distribution $\mathscr{P}(x;\theta)$, and 
\begin{align*}
    \mathcal{F}_\text{cl} (\theta) &:= \mathbb{E}_\theta\left[\left(\frac{\partial}{\partial \theta} \ln \mathscr{P}(x;\theta)\right)^2\right] \\
    &= \mathbb{V}_\theta[\ln \mathscr{P}(x;\theta)], \numberthis 
\end{align*}
is the Fisher information (FI) associated with probability distribution $\mathscr{P}(x;\theta)$.  

\begin{figure}[t]
    \centering
    \includegraphics[width=0.85\linewidth]{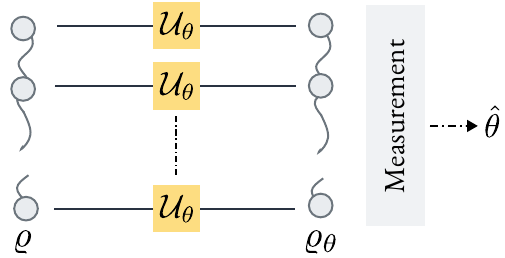}
    \caption{The general framework of quantum metrology. The probe, prepared in known initial state $\varrho$, interacts with the signal via a unitary $\mathcal{U}_\theta$. The final state $\varrho_\theta$ is then measured, from which the parameter $\theta$ is estimated. In our setting, both the state preparation and the measurement process are subject to noise. } 
    \label{fig:quantum-metrolog-protocol}
\end{figure}

\subsection{Quantum estimation theory}\label{subsec:quantum-estimation-theory}

In quantum estimation, the goal is to estimate an unknown parameter $\theta$ encoded in a quantum state $\varrho_\theta$ of the \emph{probe} after its interaction with the signal under investigation. The outcome statistics of a quantum measurement on $\varrho_\theta$ are described by the Born rule: 

\begin{equation}
    \mathscr{P}(x;\theta) = \text{Tr}\left[\Pi_x\, \varrho_\theta\right],
\end{equation}
where $\{\Pi_x\}$ is a positive operator-valued measure (POVM), i.e., a set of positive semidefinite operators satisfying $\sum_x \Pi_x = \mathbb{1}$. Each $\Pi_x$ corresponds to a measurement outcome $x$. For a given POVM $\{\Pi_x\}$, the FI can be written 
\begin{equation}\label{eq:FI-Quantum}
    \mathcal{F}_\text{cl}[\{\Pi_x\},\theta] = \sum_x \mathscr{P}(x;\theta)\left(\frac{\partial \ln \mathscr{P}(x;\theta)}{\partial \theta}\right)^2\,.
\end{equation}
The ultimate bound is achieved by maximizing the FI in \eqref{eq:FI-Quantum} over all possible measurement strategies yields the so-called quantum Fisher information (QFI) $\mathcal{F}_Q$:
\begin{equation}\label{eq:def-quantum-fi}
\mathcal{F}_Q(\theta) = \sup_{\{\Pi_x\}} \mathcal{F}_\text{cl}[\{\Pi_x\},\theta],
\end{equation}
and the generalization of the CRB 
\begin{equation}\label{eq:quantum-CRLB}
    (\Delta \hat{\theta})^2 \geq \frac{1}{N \mathcal{F}_Q(\theta) }\,,
\end{equation}
known as the quantum Cram\'{e}r-Rao bound (QCRB)\,\cite{Helstrom1969,holevo_probabilistic_2011}. 
For a initial state $\varrho_0$ undergoing unitary encoding:
$$\varrho_\theta = \exp(-iA\theta) \varrho_0 \exp(iA\theta)\,.$$
The QFI is given by:
\begin{equation}\label{eq:qfi}
    \mathcal{F}_Q[\varrho,A] = 2 \sum_{k,l} \frac{(\lambda_k - \lambda_l)^2}{(\lambda_k + \lambda_l)} |\langle k|A|l\rangle|^2,
\end{equation}
where $\lambda_k$ and $|k\rangle$ are the eigenvalues and eigenvectors of the density matrix $\varrho_0$, respectively, and the summation goes over all $k$ and $l$ such that $\lambda_k + \lambda_l>0$.
\eqref{eq:qfi} can be simplified for a pure state $\varrho_0 = |\psi\rangle \langle \psi|$:
\begin{align*}
    \mathcal{F}_Q[\varrho_0]  &= 4\left(\text{Tr}[A^2\varrho_0] - \text{Tr}[A\varrho_0]^2 \right) \equiv 4 \langle \left(\Delta A\right)^2\rangle\,. \numberthis 
\end{align*}
The QFI is maximized by the state with maximal variance of $A$, which is achieved by the equal superposition of the extremal eigenstates:
\begin{equation}
    |\psi\rangle = \frac{1}{\sqrt{2}}\left( |\lambda_\text{max}\rangle + |\lambda_\text{min}\rangle \right),
\end{equation}
where $\lambda_\text{max}$ and $\lambda_\text{min}$ are the maximum and minimum eigenvalues of $A$. The corresponding maximal QFI is $(\lambda_\text{max} - \lambda_\text{min})^2$.

\subsection{Quantum metrology with \texorpdfstring{$n$}{TEXT} probing qubits}\label{subsec:quantum metrology with n probing qubits}

In this work, we consider a local Hamiltonian acting on $n$ qubits $\mathcal{H} = \sum_{i = 1}^n Z_i/2$ which generates a unitary evolution (see \figref{fig:quantum-metrolog-protocol}):
\begin{equation}
    \mathcal{U}_\theta = e^{-i\theta \mathcal{H}} = \bigotimes_{i=1}^n e^{-i\theta Z_i/2}\,.
\end{equation}
The maximum attainable QFI is $n^2$, obtained by a GHZ state, $|\text{GHZ}_n\rangle =(|0\rangle^{\otimes n} + |1\rangle^{\otimes n})/\sqrt{2}$. This quadratic scaling of QFI is known as the HL. Classically, one choice of the optimal measurement is to measure the observable $X^{\otimes n}$, which amounts to measuring each qubit in $X$-basis and taking the parity of the outcomes. Its probability distribution function is:
\begin{align}\label{eq:parity-pdf}
   \mathscr{Q}_n({{x }};\theta) &:= \frac{1}{2}\left(1+{{x }}\cos (n\theta)\right), \quad {{x }} = \pm 1 \,.
\end{align}
The classical FI of this binary distribution is
\begin{align}\label{eq:noiseless-ghz-fi}
    \mathcal{F}_\text{cl}(\theta) &= \sum_{{{x }}=\pm1} \frac{\left(\partial_\theta \mathscr{Q}_n ({{x }};\theta)\right)^2}{ \mathscr{Q}_n ({{x }};\theta)} = n^2\,. 
\end{align}
Therefore, the parity measurement attains the QFI. 

In this work, we will work with another equivalently useful input state $|i\text{GHZ}_n\rangle := \left(|0\rangle^{\otimes n} + i|1\rangle^{\otimes n}\right)/\sqrt{2}$, which is more resilient against measurement noise (as discussed in later sections). The probability distribution associated with the parity measurement is given by:
\begin{equation}\label{eq:parity-dist-1}
   \mathscr{P}_n({{x }};\theta) := \frac{1}{2}\left(1+{{x }}\sin (n\theta)\right), \quad {{x }} = \pm 1 ,
\end{equation}
for which FI also scales as $n^2$ and therefore attains the QFI.

\subsection{Noisy quantum metrology}\label{subsec:noisy quantum metrology}

We will now consider the metrological scenario described in \secref{subsec:quantum metrology with n probing qubits} in the presence of local bit-flip noise, acting on all qubits. For now, we assume that the noise acts on the probe state before it interacts with the signal; this assumption will be relaxed later. The bit-flip channel acting on a single qubit density matrix $\rho$ is described using 
\begin{equation}\label{eq:bit-flip-channel-def}
    \Delta[\rho] = (1-p)\rho + p X\rho X\,,
\end{equation}
where $p$ is the probability of bit-flip noise and we assume $p < 1/2$. Therefore, the resulting state $\varrho_\theta$ after sensing is given by $ \mathcal{U}_\theta\left(\Delta^{\otimes n}\left[|i\text{GHZ}_n\rangle \langle i\text{GHZ}_n|\right]\right)\mathcal{U}^\dagger_\theta$. The noisy probe state $\Delta^{\otimes n}\left[|i\text{GHZ}_n\rangle \langle i\text{GHZ}_n|\right]$ is a convex combination over all possible bit-flip error patterns, i.e.
\begin{multline}
    \Delta^{\otimes n}\left[|i\text{GHZ}_n\rangle \langle i\text{GHZ}_n|\right] = \\ \sum_{b \in \mathbb{F}_2^n} p^{\text{wt}(b)} (1-p)^{n-\text{wt}(b)} X^b |i\text{GHZ}_n\rangle \langle i\text{GHZ}_n| X^b \,,
\end{multline}
where $b = (b_1,\ldots ,b_n)\in \{0,1\}^n :=\mathbb{F}^n_2$ labels the pattern of bit-flips, $\text{wt}(b)$ is the Hamming weight (number of 1s in $b$), and $X^b = X^{b_1}\otimes \cdots \otimes X^{b_n}$. For each bitstring $b\in \mathbb{F}_2^n$, define the  bit-flipped GHZ state:
\begin{align}
    |\psi_b\rangle &:= \frac{1}{\sqrt{2}} \left(|b\rangle + |\overline{b}\rangle\right) \Rightarrow \varrho_b:=|\psi_b\rangle \langle \psi_b| \,,
\end{align}
where $\overline{b} := 1^n \oplus b$. Let $\mathbb{S}_w\subset \mathbb{F}_2^n$ be the set of bitstrings with Hamming weight $w$. We now write the noisy probe state as a mixture over these $\varrho_b$'s, grouped by the number of errors $\text{wt}(b)$:
\begin{multline}
    \Delta^{\otimes n}\left[|i\text{GHZ}_n\rangle \langle i\text{GHZ}_n|\right] = \\ \sum_{w=0}^n \binom{n}{w} p^w (1-p)^{n-w}  \sum_{b\in \mathbb{S}_w} \rho_b. 
\end{multline}
We further define the \emph{magnetization number} of the state $|\psi_b\rangle$ as 
\begin{equation}
m(b):=\text{wt}(\overline{b}) - \text{wt}(b)= n - 2\text{wt}(b) \,,
\end{equation} 
which is the energy gap between the two superposed states, such that
\begin{equation}
    \mathcal{U}_\theta |\psi_b\rangle = \frac{1}{\sqrt{2}}\left(e^{-im(b)\theta /2} |b\rangle  + ie^{im(b)\theta /2}  |\overline{b}\rangle \right)\,. 
\end{equation}
The probability associated with parity measurement of the state $\varrho_\theta$ can be written as 
\begin{equation}
    \mathscr{P}({{x }};\theta) = \sum_{w=0}^n p_w\cdot \mathscr{P}_{m_w}({{x }};\theta)\,,\quad {{x }}=\pm1,
\end{equation}
where 
\begin{equation}
\label{eq:Bnwp}
p_w = \binom{n}{w} p^w(1-p)^{n-w} =: \mathcal{B}(n,w,p)
\end{equation}
is the probability of occurring weight $w$ bit-flip error and $m_w=n-2w$ is a magnetization number associated with set $\mathbb{S}_w$. We will also use $\overline{m}$ to represent the \emph{average magnetization number} over the distribution $p_w$. 
The associated FI can be calculated, and one finds 
\begin{equation}
    \mathcal{F}_\text{cl}(\theta)= 
     \frac{\left(\sum_w m_w p_w \cos(m_w\theta)\right)^2}{4\mathscr{P}(+1;\theta)\mathscr{P}(-1;\theta)}. 
\end{equation}

In this work, we will focus on parameter estimation near the point $\theta = 0$, which physically corresponds to the situation of weak signal and short-time evolution. Thus, unless otherwise mentioned, we will carry out calculations of FI at $\theta = 0$ below, which characterizes the parameter estimation precision in the vicinity of $\theta = 0$. We have 
\begin{equation}
   \mathcal{F}_\text{cl}(\theta) \big|_{\theta = 0} = \left(\sum_w m_w p_w\right)^2 = \overline{m}^2 = (1-2p)^2n^2. \label{eq:withoutQEC}
\end{equation}

This means the HL is achieved even when the bit-flip error is present. This is closely related to the fact that the input GHZ state is the logical state of the repetition code (introduced in \secref{sec:stab} in more detail) which is robust against bit-flip errors. Note that in the situation where perfect syndrome measurement is performed, the value of $|m|$ can be exactly determined and one can calculate and show that $\mathcal{F}_\text{cl}(\theta) \big|_{\theta = 0} = (1-2p)^2n^2 + 4(1-p)p n$, slightly improving \eqref{eq:withoutQEC}. 

In general, QEC has been shown useful to recover the HL in noisy quantum metrology in situations where the ``Hamiltonian-in-Kraus-span'' (HNKS) condition is satisfied~\cite{PRXQuantum.2.010343}. The case of Pauli-Z rotation and Pauli-X noise is the most well studied example~\cite{PhysRevLett.112.150802,PhysRevLett.112.150801,PhysRevLett.112.080801,ozeri2013heisenberglimitedmetrologyusing,unden2016quantum}. It satisfies the HNKS condition where the Hamiltonian (Pauli-$Z$) is not in the span of the noises (the linear span of Pauli-$X$ and identity operators). In contrast, dephasing noise or depolarizing noise in this case is fundamentally more harmful, where no-go results showed the SQL cannot be surpassed~\cite{Escher2011,4655455,Demkowicz-Dobrzański2012}.  

Even though QEC was proven useful in recovering the HL under Pauli-X noise, the proposal assumes only Pauli-X noise during the signal accumulation step, e.g. in our discussion above where Pauli-X noise acts after the state preparation and prior to the Pauli-Z rotation. Similar results carry over to situations where Pauli-X acts during or right after the Pauli-Z rotation. However, it was assumed that there are no noise in probe state preparation and measurement noise---which are ubiquitously present in practical situations. This motivates our study of fault-tolerant metrology where both Pauli-X noise and measurement noise occur in the state preparation and measurement step.

Finally, to demonstrate the harmful effect of measurement noise, we will again consider $\mathcal{H} = \sum_{i = 1}^n Z_i/2$ and assume each single-qubit measurement outcome (in the $X$-basis) is flipped independently with probability $q$. We assume no other noise, i.e., the probe state preparation ($|i\text{GHZ}_n\rangle$) and sensing are ideal. Note that the ideal distribution over $X$-basis bit-strings depends only on the product of the measured signs 
\begin{equation}\label{eq:def-parity}
    x := \prod_{j=1}^n x^{(j)} \in \{+1,-1\}\,,
\end{equation}
where $x^{(j)} = +1$ denotes $|+\rangle $ on qubit $j$, $x^{(j)}=-1$ denotes $|-\rangle$. All strings with the same product $x$ have the same probability; so the measurement effectively yields two aggregated outcomes (either $x = \pm 1$) given by \eqref{eq:parity-pdf}. An even number of flips preserves the overall parity of measured $X$-outcomes, while an odd number of flips inverts the parity. Therefore, the probability that the measurement outcome $S$ flips its value given by 
\begin{equation*}
    \sum_{i \text{ odd}} \binom{n}{i} q^i (1-q)^{n-i} = \frac{1-(1-2q)^n}{2}\,,
\end{equation*}
leading to the observed $X^{\otimes n}$ distribution
\begin{equation}\label{eq:parity-pdf-measuremnt-error}
    \mathscr{P}_n^{\text{obs}}({{x }};\theta) = \frac{1}{2}\left(1+{{x }} \eta \sin(n\theta)\right)\,,
\end{equation}
where $\eta:=(1-2q)^n$ denotes the probability bias for parity preservation. We can calculate the FI associated with \eqref{eq:parity-pdf-measuremnt-error} which leads to 
\begin{equation}
    \mathcal{F}_\text{cl}(\theta) \big|_{\theta = 0} = \frac{\eta^2 n^2\cos^2(n\theta)}{1-\eta^2 \sin^2(n\theta)} \big|_{\theta = 0} = \eta^2 n^2\,. 
\end{equation}

Therefore, the FI decays exponentially in presence of noisy measurement, which renders the metrological protocol useless when $n$ is large. 

\subsection{Stabilizer codes}\label{sec:stab}

We will work closely with stabilizer codes, particularly the repetition code. Here, we will briefly introduce concepts used in this work.

Stabilizer codes~\cite{gottesman1997stabilizercodesquantumerror} are a widely used class of quantum error-correcting codes that use the algebraic structure of the \emph{Pauli group} to protect quantum information. For $n$ qubits, the Pauli group $\mathsf{P}^n$ comprises all $n$-fold tensor products of $\left\{I,X,Y,Z\right\}$, up to global phases in $\{\pm1,\pm i\}$. Elements of $\mathsf{P}^n$ either commute or anticommute. A $\llbracket n,k \rrbracket$ stabilizer code encodes $k$ logical qubits into $n$ physical ones by defining a codespace as the simultaneous $+1$ eigenspace of an \emph{Abelian subgroup} $\mathscr{S} \subseteq \mathsf{P}^n$, called the stabilizer group, which excludes $-I^{\otimes n}$ to ensure a non-trivial code. Typically, $\mathscr{S}$ has $(n - k)$ independent generators, similar to parity checks in classical codes, yielding a codespace of dimension $2^k$.

The centralizer of the stabilizer group is defined as 
\begin{equation}
    C(\mathscr{S})=\left\{P\in \mathsf{P}^n\ |\ [P,S]=0\ \forall S\in \mathscr{S}\right\}
\end{equation}
Logical operators $\mathscr{L}\subseteq \mathsf{P}^n$ are elements of $C(\mathscr{S})$ that are not in $\mathscr{S}$, i.e. the nontrivial cosets in the quotient group $C(\mathscr{S})/\mathscr{S}$.

The code distance $d$ is the minimum weight of such a nontrivial logical operator:
\begin{equation}
    d = \min\left\{\text{wt}(P)\ | \ P\in C(\mathscr{S})\setminus \mathscr{S}\right\}\,,
\end{equation}
quantifying the code's ability to correct up to $\lfloor(d-1)/2\rfloor$ qubit errors. Since any quantum error can be decomposed into Pauli operators, we consider a discrete error set $\mathsf{E} \subset \mathsf{P}^n$. The code can correct $\mathsf{E}$ if for all $E_1,E_2 \in \mathsf{E}$, the operator $E_1^\dagger E_2$ either: (a) anticommutes with at least one element of stabilizer (yielding a distinguishable syndrome), or (b) belongs to the stabilizer group $\mathscr{S}$ (acting trivially on the codespace). Formally~\cite{PhysRevA.55.900,PhysRevLett.84.2525},
\[
E_1^\dagger E_2 \notin C(\mathscr{S}) \quad \text{or} \quad E_1^\dagger E_2 \in \mathscr{S},
\]
This condition ensures that distinct errors produce distinct syndromes, enabling reliable recovery. 

\begin{figure}
    \centering
    \includegraphics[width=0.85\linewidth]{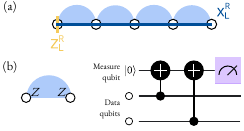} 
    \caption{ (a) Illustration of the distance-5 repetition code. Data qubits are represented by open circles. (b) The quantum circuit for measuring $Z$ syndrome.}
    \label{fig:stabilizer-codes}
\end{figure}

\subsubsection{Repetition code}\label{subsubsec:repetition-code}

A repetition code is the simplest example of stabilizer code~\cite{PhysRevA.52.R2493} (see \figaref{fig:stabilizer-codes}). It encodes a single logical qubit into $n$ physical qubits (referred to as \textit{data qubits}) or $\llbracket n,1 \rrbracket$ by redundantly copying the logical information. The codespace is stabilized by $(n-1)$ commuting operators of the form:
\begin{equation}\label{eq:repetition-stabilizers}
    \mathscr{S}^\textsf{R} = \left\langle Z_1Z_2,Z_2Z_3,\ldots,Z_{n-1}Z_n \right\rangle\,,
\end{equation}
where $\langle \cdot \rangle$ denotes the group generated by these operators, i.e., the set of all products of these generators (including the identity). These generators enforce that all qubits are in the same $Z$-basis state, either all $|0\rangle $ or all $|1\rangle$. Thus, the logical basis states are $|0\rangle^\textsf{R}_\textsf{L} = |0\rangle^{\otimes n}$ and $|1\rangle^\textsf{R}_\textsf{L} = |1\rangle^{\otimes n}$. The logical Pauli operators are: $Z^\textsf{R}_\textsf{L} = Z_i$ (for any $i$) and $X^\textsf{R}_\textsf{L} = X^{\otimes n}$. Because any single bit-flip error anticommutes with exactly one stabilizer generator and changes the syndrome, the code can detect up to $\lfloor (n-1)/2 \rfloor$ such errors. This form of repetition code is useful for correcting bit-flip errors, but not phase-flip errors, since all stabilizers are diagonal in the $Z$-basis. 

The set of physical qubits (referred to as \textit{measure qubits}) are used to measure $Z$ syndrome ($Z$-stabilizers in \eqref{eq:repetition-stabilizers}) as shown in \figbref{fig:stabilizer-codes}.

\subsubsection{The minimum weight matching decoder}\label{subsubsec:the decoding problem}

Given an error $E\in \mathsf{E}$, the decoder uses its syndrome $\sigma(E)$ to select a correction operator $R\in \mathsf{P}^n$ to apply to the corrupted state $E|\phi\rangle$. Decoding is successful if $RE\in \mathscr{S}$, and fails otherwise. If $R$ is consistent with the syndrome, then $RE\in C(\mathscr{S})$. A logical error occurs when $RE\in C(\mathscr{S})\setminus \mathscr{S}$. The decoder's task is to return any element of the coset $[E]:=\{ES\quad \forall S\in \mathscr{S}\}$. 

\begin{figure}
    \centering
    \includegraphics[width=0.85\linewidth]{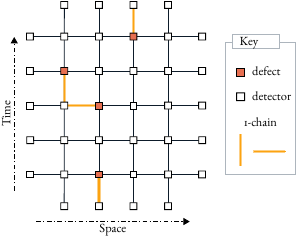}
    \caption{The MWPM decoding problem for a distance-5 repetition code (with measurement errors). The $X$-error matching graph is shown with $5$ rounds of syndrome measurement. Each detector represents the difference syndrome i.e., the difference (modulo 2) between the syndrome measurement in time step $t$ and $(t-1)$, and ensures that any single measurement error results in two syndrome defects. We assume all syndrome outcomes can be flipped with probability $\mathsf{q}$, even in the last round of the syndrome measurement.}
    \label{fig:rep-matching}
\end{figure}

In this work, we will use a minimum weight decoder which identifies the minimum weight error consistent with the syndrome~\cite{10.1063/1.1499754,deMartiiOlius2024decodingalgorithms,PhysRevA.86.032324}. While its performance might be inferior to the maximum likelihood decoding~\cite{7097029}, the minimum-weight perfect matching (MWPM) decoder provides an efficient method for solving the decoding problem for certain important code families, including the repetition code and the surface code (with either $X$- or $Z$-type errors)~\cite{Higgott2025sparseblossom,10.1145/3505637, 10313859}.

We will now consider the problem of decoding a Pauli $X$ error in $\mathsf{E}_X=\{I,X\}^n$ (assuming no measurement errors) for a stabilizer code~\cite{deMartiiOlius2024decodingalgorithms,Higgott2025sparseblossom,10.1145/3505637}. The MWPM decoder begins with defining a matching graph $\mathcal{G}=(\mathcal{V},\mathcal{E})$ (also known as decoding graph or detector graph). Each node $v \in \mathcal{V}$ represents a $Z$-type check or a \emph{detector}. Each edge $e\in \mathcal{E}$ is a set of one or two detector nodes, representing an $X$-error that flips this set of detectors. We can be decomposed the edge set as $\mathcal{E}=\mathcal{E}_1\cup\mathcal{E}_2$, where $\forall e\in\mathcal{E}_1:|e|=1$ (refer as half-edge) and $\forall e\in\mathcal{E}_2:|e|=2$, where $|e|$ represents the cardinality of the edge $e$. For a half-edge $e=(u,)\in \mathcal{E}_1$ corresponds to an error for which there is only one anti-commuting $Z$-checks, in other words, there is a single detector $u\in \mathcal{V}$ that can flip. On the other hand, a regular edge $e=(u,v)\in \mathcal{E}_2$ corresponds to an error that can flip two detectors $(u,v)\in \mathcal{V}$. For a half-edge $(u,)\in \mathcal{E}$ we introduce a virtual node $v_b$.

An error $E\in \mathsf{E}_X$ corresponds to a subset of edges called a $1$-\textit{chain}, and the subset of nodes $\mathcal{D}\subseteq \mathcal{V}$ corresponding to $Z$-checks that anticommute with $E$, referred to as \textit{defects} (i.e., detectors that register a negative syndrome outcome). Minimum weight decoding of $X$ errors then corresponds to finding the shortest $1$-chain whose endpoints have the defect nodes and/or virtual nodes\footnote[2]{If each qubit is exposed to a different error probability, we assign the corresponding edge the weight $w_i = \ln((1-p_i)/p_i)$, where $p_i$ is the associated error probability.}

When syndrome measurements themselves are noisy, measurements are repeated $O(d)$ times, where $d$ is the distance of the code~\cite{10.1063/1.1499754}. In this case, we define a detector as a change in the outcome of a stabilizer measurement between consecutive rounds, and a defect as a detector that registers a nontrivial change. For the repetition code, repeated noisy syndrome measurements give rise to a 2D matching graph, as illustrated in \figref{fig:rep-matching}.

\section{Overview of protocol}\label{sec:overview-of-protocol}

In this section, we provide an overview of our protocol before describing each step in detail. The protocol consists of three main stages: (1) preparation of the probe state, which includes a fault-tolerant preparation of the GHZ state $\left(|0\rangle^{\otimes n} + |1\rangle^{\otimes n}\right)/\sqrt{2}$ and a phase gate $S = \ket{0}\bra{0} + i\ket{1}\bra{1}$ afterwards, applied uniformly randomly on one of the $n$ probe qubits, (2) interaction of the probe with the signal $\mathcal H = \sum_{i} Z_i/2$, and (3) measurement of the parity operator $X^{\otimes n}$. 

We will calculate and plot in \secref{sec:performance-fi} the FI of our protocol at $\theta = 0$ and show it approaches the HL when the physical error rates are below certain thresholds.

In this work, we adopt a phenomenological noise model that is equivalent to the circuit-level noise mode defined in Def.\,\ref{def:noise model}. We allow for effective errors associated with the syndrome measurement process. At the phenomenological level, we model a faulty syndrome outcome by assuming that each syndrome measurement value is independently flipped with probability $\mathsf{q}$. In addition, we introduce an effective qubit error rate $\mathsf{p}$, which accounts for accumulated Pauli-$X$ errors suffered by the data qubits during a syndrome measurement round. Both $\mathsf{q}$ and $\mathsf{p}$ are treated as independent effective error parameters, capturing the net effect of circuit-level faults during syndrome extraction without explicitly modeling the underlying measurement circuit. A summary of all error parameters and their physical interpretations is provided in Appendix\,\ref{appendix:notation}.

Although we describe errors using a phenomenological model for simplicity, our protocol is fully compatible with circuit-level noise. In a circuit-level description, faulty syndrome measurements arise from Pauli-$X$ errors occurring at any location in the syndrome-extraction circuit, including measure qubit preparation, two-qubit gates, and ancilla readout. The effective parameter $\mathsf{q}$ (likewise $\mathsf{p}$) can therefore be viewed as incorporating all such contributions (see Appendix\,\ref{appendix:syndrome measurement}). Our threshold results and scaling behavior remain unchanged under this more realistic noise model.

Importantly, there are no hook errors in our setting. A hook error refers to a correlated multi-qubit data error generated when a single fault in a syndrome-extraction circuit propagates from the measurement qubit to multiple data qubits. In our protocol, we consider only Pauli-$X$ noise, and the stabilizers are measured using a circuit in which faults on the measurement qubit do not propagate back to the probe qubits. As a result, a single physical fault cannot produce correlated multi-qubit $X$ errors, and hook errors are absent.


We assume all errors are independent and uncorrelated in space and time. Finally, we do not consider dephasing noise, as previous no-go results have established that the Heisenberg limit is unattainable in that case.


The procedure begins with preparing an $n$-qubit GHZ state. We do this by initializing the product state $|+\rangle^{\otimes n}$ --- due to state preparation errors, some of them are prepared in the $|-\rangle$ state. The state preparation process can be thought of as $Z$-error after $|+\rangle $ initialization, such an error commutes with stabilizer measurements and cannot be corrected. To suppress these errors, we perform repeated non-destructive measurements of the Pauli-$X$ operator on each probe qubit using ancilla-assisted measurement circuits. The measurement outcomes are classically processed using majority voting to infer the likely eigenvalue of $X$. If the inferred outcome is $-1$, we apply a corrective Pauli-$Z$ operation to flip the state back to $|+\rangle $. As shown in Appendix\,\ref{appendix:probe initialization error}, only $O(\ln n )$ such measurements are sufficient to suppress the state preparation error. After initialization, we perform $O(\ln n)$ rounds of syndrome measurements of the stabilizer generators in $\mathscr{S}^\mathsf{R}$, followed by decoding operations.\footnote[4]{In general, as we discussed in Sec.\,\ref{subsubsec:the decoding problem}, one require to perform $O(n)$ round of syndrome measurements, however, as we show in Sec.\,\ref{sec:state preparation}, we only require $O(\ln n)$ measurements due to presence of undetectable local error.}

Note that the additional phase $i$ in the probe state $|\psi\rangle$ is added  to guarantee the output parity distribution is $\{1/2,1/2\}$ at $\theta = 0$, which is the most robust choice against measurement noise. Furthermore, we perform the phase gate on a qubit picked uniformly randomly from all qubits to make sure the associated gate noise is applied uniformly randomly on all qubits. In this step, the $S$ gate may turn an $X$ error before its application into an $Z$ error after its application, which cannot be corrected. However, the scaling of FI will not be affected since the $Z$ error only acts on one qubit and is not extensive. 

The probe state interacts with the signal via the Hamiltonian $\mathcal{H} = \sum_{i} Z_i/2$ which generates the unitary evolution $\mathcal{U}(\theta) = \exp\left(-i\theta \mathcal{H}\right)$. We assume the system-probe interaction is instantaneous, so that no bit-flip error occurs during the sensing process. 

Assuming perfect state preparation, after the quantum signal acts on the probe state $|\psi\rangle$, the state evolves to 
\begin{equation}
    |\psi_\theta\rangle^\mathsf{R}_\mathsf{L} = \frac{1}{\sqrt{2}}\left(e^{-in\theta}|0\rangle^\mathsf{R}_\mathsf{L} +ie^{in\theta}|1\rangle^\mathsf{R}_\mathsf{L}\right). 
\end{equation}
To measure the logical parity operator $X^\mathsf{R}_\mathsf{L}$ in a fault-tolerant manner, we sequentially perform non-destructive, ancilla-assisted measurements of the single-qubit operators $X_1,\ldots ,X_n$. Since individual measurement outcomes are faulty, each $X_i$ is measured multiple times using fresh ancilla qubits while preserving the post-measurement eigenstate. The resulting outcomes are classically processed using majority voting. We show that, below a threshold error rate, the probability of an incorrect logical parity measurement decays exponentially with the number of repetitions, which scales as $O(\ln n)$.

\section{State preparation}\label{sec:state preparation}

In this section, we describe the preparation of the probe state as the logical $|+\rangle^\textsf{R}_\textsf{L}$ state of the repetition code. Further discussion on the additional phase gate afterwards will be deferred to \secref{sec:performance-fi}. The procedure involves initializing qubits in an eigenstate of the logical operator $X_\textsf{L}^\textsf{R}$, repeatedly measuring the stabilizers, and decoding the resulting syndrome history to apply corrections.

The preparation of an $n$-qubit GHZ state is equivalent to the preparation of the logical state $|+\rangle^\textsf{R}_\textsf{L}$ of the $n$-qubit repetition code. Assuming perfect syndrome measurements, the preparation is done in the following manner. We begin with the perfect unentangled product state $|+\rangle^{\otimes n}$. In practice, state-preparation errors introduce a small number of Pauli-$Z$ errors; however, since such errors commute with the stabilizer measurement of the repetition code, they do not affect the preparation procedure itself. Their impact is instead accounted for later when we evaluate the Fisher information of the protocol in Sec.\,\ref{sec:performance-fi}.

The initial state $|+\rangle^{\otimes n}$ is an eigenstate of the logical operator $X^\textsf{R}_\textsf{L}$, but it is not stabilized by the stabilizers $\mathscr{S}^\mathsf{R}$. Next, we measure all stabilizer generators in $\mathscr{S}^\mathsf{R}$. Since the stabilizers commute with the logical operator $X^\textsf{R}_\textsf{L}$, this measurement does not disturb the logical value. The measurement outcomes result in defects ($-1$ value of syndrome measurement outcome). To restore the state to the code space, we must eliminate all such defects. We select an arbitrary one-chain (a collection of adjacent qubit locations) whose boundary coincides with the positions of the detected defects. Applying $X$ operators along this chain implements a correction that removes the defects. Although this $X$ correction may apply a logical $X^\textsf{R}_\textsf{L}$ operator if the chain is homologically nontrivial, this has no effect, as the state is in the $+1$ eigenspace of $X^\textsf{R}_\textsf{L}$. The result is a state stabilized by both the code’s stabilizers and the logical operator $X^\textsf{R}_\textsf{L}$, i.e., the desired logical state $|+\rangle^\textsf{R}_\textsf{L}$.

\begin{figure}
\centering 
\includegraphics[width=\linewidth]{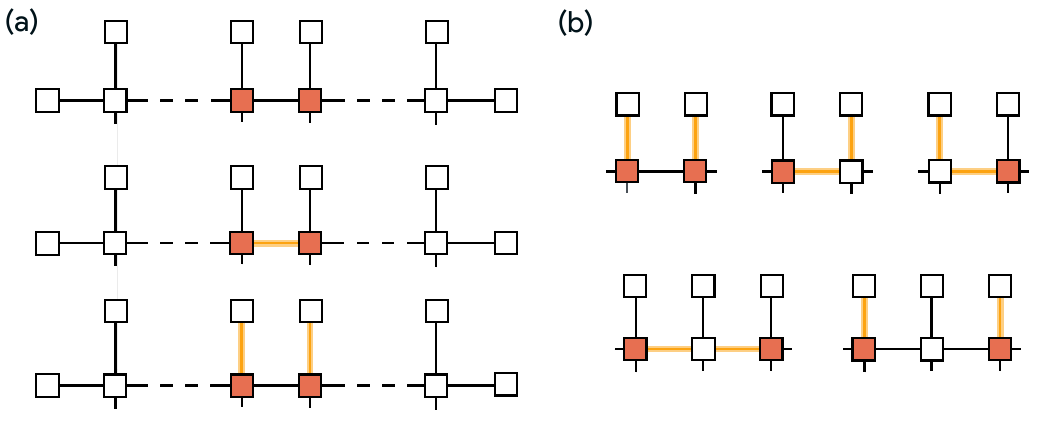}
\caption{
(a) The top panel shows the observation of two consecutive defects. These defects may be caused by a single-qubit error (middle panel), which occurs with probability $\mathsf{p}$, or by two consecutive measurement errors (bottom panel), which occurs with probability $\mathsf{p}^2$. (b) All possible error mechanisms that occur with probability $\mathsf{p}^2$---each leading to a failure of the matching algorithm. }
\label{fig:decoding_state_prep}
\end{figure}

In the presence of syndrome measurement errors (with effective error probability $\mathsf{q}$), the above preparation procedure can produce long open chains of $X$ errors (with effective probability $\mathsf{p}$ at all locations). We will assume $\mathsf{p}=\mathsf{q}$ for simplicity, however, our protocol does not require and is independent of this choice. Generally to suppress these errors, one needs to perform $n$ rounds of stabilizer measurement. However, in our case, (discussed in detail below) $O(\ln n)$ rounds of measurement suffices since there are local errors that dominate at small $\mathsf{p}$. Next we use syndrome decoding over the resulting space-time graph to identify and apply the appropriate correction. In \figref{fig:rep-matching}, we show a decoding graph for the repetition code. The horizontal and vertical edges are space-like and time-like edges, which represent $X$ error and measurement error mechanism, respectively. The MWPM decoder gives a set of space-like and time-like errors corresponding to the syndrome. Any time-like matching to virtual detector corresponds to a measurement error in the last syndrome measurement. We use these matchings to correct the syndrome measurement and apply the correction to obtain the GHZ state $|+\rangle^\textsf{R}_\textsf{L}$.

We can qualitatively understand when the matching algorithm fails. We note that if we can predict the syndrome of the last step faithfully, then we can obtain a correction to perfectly recover the state $|+\rangle^\textsf{R}_\textsf{L}$. Therefore, to see when the matching algorithm might fail, we particularly focus on the defect matching corresponding to the last step. We can qualitatively understand the types of errors that lead to logical errors. Consider the case where we have two consecutive defects, shown in \figaref{fig:decoding_state_prep}. These defects could arise either from a single qubit error or two consecutive measurement errors. Since a single-qubit error occurs with probability $\mathsf{p}$, whereas two measurement errors occur with probability $\mathsf{p}^2$, the matching algorithm will typically misidentify the error as a single-qubit error with probability of order $\mathsf{p}^2$.
Such a misidentification leads to a single bit-flip error in our GHZ state. The average number of bit flips from this error mechanism can be calculated roughly by adding up such errors across the entire chain, giving a contribution of $\Theta(n\mathsf{p}^2)$. In \figbref{fig:decoding_state_prep}, we illustrate all possible $\mathsf{p}^2$ error mechanisms. Each of these mechanisms contributes $\Theta(n\mathsf{p}^2)$ to the total average number of bit flips.

\begin{figure}
    \centering
    \includegraphics[width=0.95\linewidth]{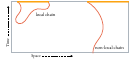}
    \caption{Illustration showing local chain that has end points in the same time-like boundary, while non-local chain that has end points in the opposite time-like boundary. The yellow lines show the error associated with the respective chain (shown in orange).}
    \label{fig:error-chains}
\end{figure}

\begin{figure*}[t]
    \centering
    \includegraphics[width=0.95\linewidth]{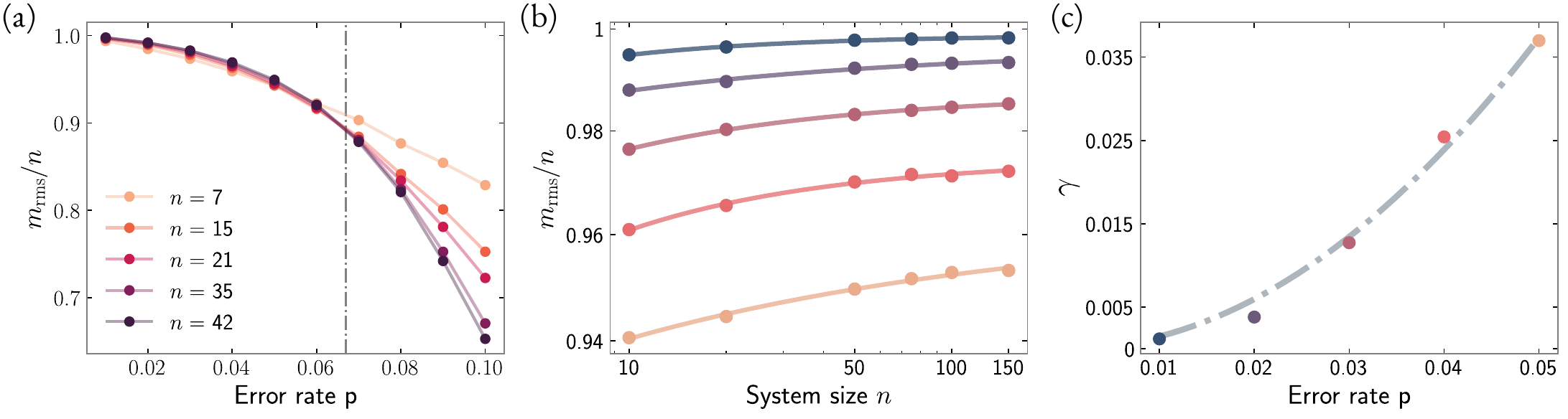}
    \caption{
    Numerical simulation of state preparation, where $\lceil 5\ln n\rceil$ rounds of repeated syndrome measurements are performed. (a) The normalized root mean square magnetization $m_\text{rms}/n$, with error rates $\mathsf{p}=\mathsf{q}$. The dotted line shows the threshold value ($p^{(s)}_\text{th} \approx 0.067$) found by calculating crossover at finite system-size. 
    Below threshold, the normalized rms magnetization number is approximately $1-2\alpha \mathsf{p}^2$ for some constant $\alpha$. (b) The same quantity computed below threshold for various $\mathsf{p}$, asymptotically approach $1-2\alpha \mathsf{p}^2$. The line shows a fit of the form $1-\gamma-c/n^\nu$ with parameters ($\gamma,c,\nu$). From top to bottom, colors correspond to $\mathsf{p} = 0.01,0.02,0.03,0.04,0.05$ respectively. (c) The fitted value of $\gamma$ shown for the same $\mathsf{p}$ values as (b), showing that $\gamma$ is of order $\mathsf{p}^2$.}
    \label{fig:state-prep-err}
\end{figure*}

Next, we show that $O(\ln n)$ repeated syndrome measurements are sufficient to suppress non-local errors. Let us assume that we perform $a$ repeated syndrome measurements. What is the error probability $p_\text{nl}$ associated with non-local error that connects the two time-like boundaries? These errors happen when the actual error chain $e$ and the minimum weight chain $e_\text{main}$ form a homologically nontrivial, i.e., the path $e+ e_\text{min}$ connects the opposite boundary in the syndrome graph. It follows that there must be at least one path of $\ell \geq a$ detectors connecting opposing boundaries containing at least $\ell_\text{min} \leq \lceil \ell/2\rceil $. Given a particular chain of length $\ell$, the probability of at least $\lceil \ell/2 \rceil $ of its lines being associated with errors is 
\begin{align*}
    \sum_{i= \lceil \ell/2\rceil}^\ell \binom{\ell}{i} \mathsf{p}^i(1-\mathsf{p})^{\ell-i} &\leq  \sum_{i= \lceil \ell/2\rceil}^\ell \binom{\ell}{\lceil \ell/2\rceil}  \mathsf{p}^i \,, \\
     &= \binom{\ell}{\lceil \ell/2\rceil} \mathsf{p}^{\lceil \ell/2\rceil} \sum_{i = 0}^{\ell - \lceil \ell/2\rceil} \mathsf{p}^{i} \\
     &\leq \binom{\ell}{\lceil \ell/2\rceil} \mathsf{p}^{\lceil \ell/2\rceil} \frac{1}{1-\mathsf{p}} \\ 
     & \leq  \mathsf{p}^{\lceil \ell/2\rceil}  \sum_{i=0}^\ell   \binom{\ell}{\lceil i \rceil} \\
     &= 2^\ell \mathsf{p}^{\lceil \ell/2\rceil} \,. 
\end{align*}
Furthermore, the number of paths of length $\ell$ can be upper bounded by first choosing one boundary as the set of possible starting points, giving $(n-1)$ choices.  

Each detector is connected to at most four neighbors, and to avoid backtracking, each step can proceed to at most three of them. After selecting the initial point, there are $\ell -1$ directional choices, leading to no more than $(n-1)3^{\ell-1}$ possible chains of length $\ell$. The probability of a logical error $\wp$ is therefore upper bounded by
\begin{align*}
    \mathsf{p}_\text{nl} &\leq  \sum_{\ell = a}^\infty n 3^{\ell-1} 2^{\ell} \mathsf{p}^{\lceil \ell/2\rceil }\\
    &= \frac{n}{3} \sum_{\ell = a}^\infty 6^\ell \mathsf{p}^{\lceil \ell/2\rceil } \\
    &\leq \frac{n}{3} (6\sqrt{\mathsf{p}})^{a}\sum_{\ell = 0}^\infty (6\sqrt{\mathsf{p}})^{\ell}\\
    &= \frac{n}{3} (6\sqrt{\mathsf{p}})^{a} \frac{1}{1-6\sqrt{\mathsf{p}}}\,.
\end{align*}
Let us choose $n=2^a$ so that
\begin{equation}
\label{eq:step2 measurement}
    \mathsf{p}_\text{nl} \leq \frac{1}{3(1-6\sqrt{\mathsf{p}})} \left(144\mathsf{p} \right)^{a/2}\,.
\end{equation}

Assume that $\mathsf{p}_w$ denotes the probability of having a weight-$w$ bit-flip error. Then the noisy GHZ state has the average magnetization number $\overline{m} = \sum_w \mathsf{p}_w m_w$. When the noisy GHZ state interacts with the quantum signal, it acquires on average a phase $\overline{m} \theta$. As we will discuss in more detail in \secref{sec:performance-fi}, unlike the noisy case discussed in \secref{subsec:noisy quantum metrology}, the FI in our protocol depends not on the average magnetization $\overline{m}$, but rather on its second moment or mean-square magnetization $\overline{m^2}\equiv m^2_\text{rms}$.

We will now calculate $m_\text{rms}$ in the case where there are two error mechanisms --- (a) local errors that occur with order $\mathsf{p}^2$ produce a single bit-flip errors, and (b) non-local time-like error that occurs with probability $\mathsf{p}_\text{nl}$ and produces $O(n)$ bit-flip error. For the latter, it suffices to assume that any error results in zero magnetization. We will assume that the two error mechanisms act independently such that $\mathsf{p}_w = \mathsf{p}^{\text{l}}_w+ \mathsf{p}^{\text{nl}}_w$, where $\mathsf{p}^{\text{l}}_w$ and $\mathsf{p}^{\text{nl}}_w$ are probability distributions associated with two error mechanisms. Therefore,
\begin{align*}
    \mathbb{E}[w] &= \mathbb{E}_\text{l}[w] + \mathbb{E}_\text{nl}[w] =  \alpha n \mathsf{p}^2 + \frac{n}{2} \mathsf{p}_\text{nl} \,, \\ 
    \mathbb{E}[w^2] &= \mathbb{E}_\text{l}[w^2] + \mathbb{E}_\text{nl}[w^2] = \mathbb{E}_\text{l}[w^2] +  \frac{n^2}{4}\mathsf{p}_\text{nl} \,,
\end{align*}
where $\mathbb{E}_\text{l}[\cdot ]$ and $\mathbb{E}_\text{nl}[\cdot ]$ represents the expectation value with respect to $\mathsf{p}^{\text{l}}_w$ and $\mathsf{p}^{\text{nl}}_w$, respectively.

Let $b_1,b_2,\ldots ,b_n$ be binary random variables indicating bit-flip errors on each qubit, and define the error weight $w = \sum_i b_i $. We then have 
\begin{align*}
    &\mathbb{E}_\text{l}\!\left[w^2\right] \\
    &=\left( \sum_i \mathbb{E}_\text{l}[b_i] \right)^2 + \sum_i \mathbb{V}_\text{l}\left[b_i\right] + 2 \sum_{i<j} \text{cov}\!\left(b_i,b_j\right)\,, 
    \numberthis \label{eq:w2}
\end{align*}
where $\text{cov}(b_i,b_j) \equiv \mathbb{E}_\text{l}\!\left[\left(b_i - \mathbb{E}_\text{l}[b_i] \right)\left(b_j - \mathbb{E}_\text{l}[b_j]\right)\right]$ is covariance associated with random variables $b_i,b_j$. The correlations between errors are expected to be short-ranged, such that $\text{cov}\!\left(b_i,b_j\right) = \text{cov}\!\left(|i-j| \right)$, which has a finite correlation length and decays exponentially with distance $|i-j|$. Therefore, the first term in \eqref{eq:w2} will be the dominant term with order $O(n^2)$ and the remaining two terms will be $O(n)$. Therefore,
\begin{align*}
    m_\text{rms}^2 &=  \mathbb{E}[(n-2w)^2]  \approx n^2 \left[ (1-2\alpha \mathsf{p}^2)^2 + \frac{\mathsf{p}_\text{nl}}{2} \right]\,.
\end{align*}
Since below threshold $p_\text{th} = 1/144$ and in limit $n\rightarrow \infty $, $\mathsf{p}_\text{nl} \ll \mathsf{p}^2$, we can write 
\begin{equation}\label{eq:rms approx av}
    \frac{m_\text{rms}}{n} \approx \frac{\overline{m}}{n} \approx 1 - 2\alpha \mathsf{p}^2, 
\end{equation}
where the approximation error decays with $n$ when $p$ is sufficiently small. 
\eqref{eq:rms approx av} does not explicit assume the form of the error distribution, therefore, remains valid under general short-ranged correlated noise with finite correlation length. If we further assume that error occurs approximately independently across qubits so that the error-weight distribution is roughly binomial, i.e.,  $\mathsf{p}_w \simeq \mathcal{B}(n,w, \alpha \mathsf{p}^2)$. Then 
\begin{equation}\label{eq:rms binomial}
    \frac{m_\text{rms}}{n} = 1 - 2\alpha \mathsf{p}^2 + O\left( \frac{\mathsf{p}^2}{n}\right)\,,
\end{equation}
for small $\mathsf{p}$.

\vspace{0.1cm}

\noindent \textit{Numerics}.--- The left panel of \figaref{fig:state-prep-err} shows the normalized root mean square (rms) magnetization, $m_\text{rms}/n$. We observe a threshold behavior such that below the critical error rate $p^{(s)}_\text{th}$, it approaches $1 - 2 \alpha \mathsf{p}^2$. From finite-size scaling, we estimate the threshold to be
\begin{equation}
    p^{(s)}_\text{th} = 0.067\,.
\end{equation}
To explicitly demonstrate this behavior, we compute normalized rms magnetization number for varying system sizes $n$ below the threshold (see \figbref{fig:state-prep-err}). The data are fitted using the curve $1-\gamma-c/n^\nu$ with fitting parameters $(\gamma,c,\nu)$. Consequently, $1-m_\text{rms}/n\rightarrow \gamma$ as $n\rightarrow \infty$, where $\gamma$ exhibits quadratic scaling of order $\mathsf{p}^2$, as shown in \figcref{fig:state-prep-err}. The threshold behavior remains unchanged when $\mathsf{p}\not = \mathsf{q}$; a detailed analysis of this more general case is provided in Appendix\,\ref{appendix:state prep p not q}.

\section{Measurement}\label{sec:measurement}

After the interaction with the signal, the resultant logical repetition state is given by (in the ideal case of perfect state preparation)
\begin{equation}
    |\psi_\theta\rangle^\mathsf{R}_\mathsf{L} = \frac{1}{\sqrt{2}}\left(e^{-in\theta/2}|0\rangle^\mathsf{R}_\mathsf{L} +ie^{in\theta/2}|1\rangle^\mathsf{R}_\mathsf{L}\right)\,,
\end{equation}
for which we need to measure $X^\mathsf{R}_\mathsf{L}$ or $X^{\otimes n}$. To perform the measurement in a fault-tolerant manner, we measure the operators $X_1, X_2,\ldots ,X_n$ using sequentially applying CNOT gates acting on  ancilla qubits (initialized in $|+\rangle$ state) and data qubits and then measuring ancilla qubits in $X$-basis as shown in \figref{fig:meas-circuit}. After the measurement, the data qubits results in the product state $\bigotimes_j |x^{(j)}\rangle $ whose parity (\eqref{eq:def-parity}) has the probability distribution given in \eqref{eq:parity-dist-1}.

\begin{figure}
    \centering
    \includegraphics[width=0.95\linewidth]{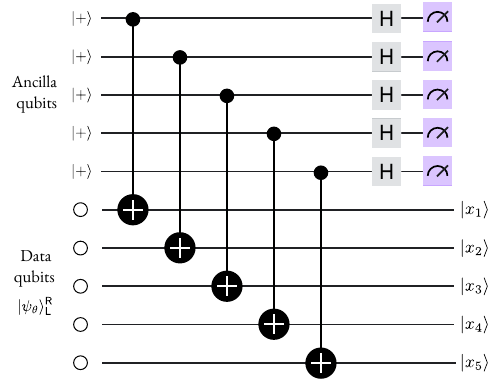}
    \caption{The quantum circuit for measuring the $X^\mathsf{R}_\mathsf{L}$ on the probe state. After the ancilla qubits are measured, the probe qubits collapse to the product state $\bigotimes_i |x_i\rangle $ with the probability distribution \eqref{eq:parity-dist-1}. }
    \label{fig:meas-circuit}
\end{figure}

The logical operator $X^\mathsf{R}_\mathsf{L}$ commutes with bit-flip errors and any $Z$-error does not traverse to data qubit, therefore, the data qubits are not affected by any qubit errors in the measurement circuit. Although the ancilla measurement result itself can flip due to either measurement or qubit error. Let us assume that $q_x$ is the error probability associated with measurement of any $X_i$ in a single measurement process. How many times should we repeat the measurement such that error probability associated with logical operator $X^\mathsf{R}_\mathsf{L}$ is sufficiently small? As shown in Appendix\,\ref{appendix:probe initialization error}, suppose we do $r$ rounds of measurement to obtain $x_i$, the probability the majority vote fails is 
\begin{align*}
    \sum_{k= \lceil r/2\rceil}^r \binom{r}{k} q_x^k(1-q_x)^{r-k} 
     &\leq 2^r q_x^{r/2+1}. 
\end{align*}
Then, the probability of failure $\wp_\text{meas}$ in computing the complete set of $X$-measurements is given by 
\begin{equation}
    \wp_\text{meas} \leq 1 - \left( 1- 2^r q_x^{r/2+1} \right)^n\,.
\end{equation}
We therefore identify a threshold for the logical $X_\mathsf{L}^\mathsf{R}$ measurement error rate,
\begin{equation}
    q_m^{(\text{th})} = \frac{1}{4}\,,
\end{equation}
below which the logical measurement error $\wp_\text{meas}$ can be made arbitrarily small by increasing the number repetitions.

When $q_x< q_m^\text{(th)}$, we can pick e.g. $r = -\frac{\ln({q_x^2n^4})}{\ln(4 q_x)}$, then 
\begin{equation}
\label{eq:step1 measurement}
\wp_\text{meas} \leq 1 - \left(1 - 1/ n^2\right)^{n},    
\end{equation}
which approaches $0$ for large $n$.

\begin{figure}[t]
    \centering
    \includegraphics[width=0.9\linewidth]{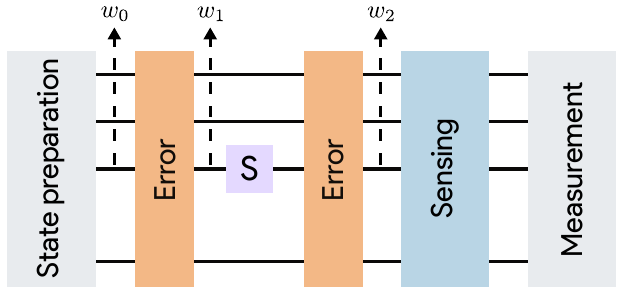}
    \caption{Schematic showing errors considered for the FI calculation. State preparation block represents the first part of state preparation considered in \secref{sec:state preparation}, which leaves some errors after the error correction. Further bit-flip errors can occur with probability $p$ on each qubit before and after the phase gate, which is applied uniformly randomly on one of the $n$ probe qubits. The measurement block corresponds to probe measurement (discussed in \secref{sec:measurement}) with error $\wp_\text{meas}$.}
    \label{fig:fi-calculation}
\end{figure}

\section{Performance of FI}\label{sec:performance-fi}

In this section, we will compute and plot the FI of the probability distribution associated with the final probe measurement. In \figref{fig:fi-calculation}, we show the overview of our protocol that we will use to calculate the FI. The state preparation involves the procedure discussed in \secref{sec:state preparation}, which leaves the residue of errors after the correction. There are further errors occurring before and after the phase gate, applied on the random qubit and during the wait step for the rest of the qubits. After the state undergoes sensing procedure, the error occur during the probe measurement procedure (discussed in \secref{sec:measurement}). Although as we discussed, the bit-flip error does not affect the measurement outcome, so the only contribution comes from the measurement errors.


Suppose that $P(w_0,w_1,w_2)$ represents the probability of error patterns, where $w_0, w_1$, and $w_2$ are weights of the bit-flip errors after the error correction step, before the phase gate, and after the phase gate, respectively. An important error event (referred to here as a phase error) is one which leaves an $X$-error before the phase gate, which flips the final parity outcome i.e. changing the final probability distribution from $\mathscr{P}_m({{x }};\theta)$ to $\mathscr{P}_m(-{{x }};\theta)$ for some $m$. Likewise, any error during the probe qubit initialization (i.e. when a qubit prepared in the $|+\rangle$ state is erroneously prepared in the $|-\rangle$) also flips the final parity outcome. Let us $\wp_\text{prep}$ is the probe preparation error. Since there are $w_1$ errors before the phase gate, the probability of a phase error is $p_s := w_1/n$. The final probability distribution of the probe measurement is given by 
\begin{align*}\label{eq:logical-dist}
    \mathscr{P}({{x }};\theta) &= \sum_{\substack{w_0 w_1w_2}} P(w_0,w_1,w_2) \left[(1-p_{\oplus})\cdot \mathscr{P}_{m_{w_2}}({{x }};\theta) \right. \\
    &\quad  \left. + p_\oplus\cdot \mathscr{P}_{m_{w_2}}(-{{x }};\theta) \right]\,, \numberthis 
\end{align*}
where $p_\oplus := \text{Pr(odd number of flips)}$ denotes the effective flip probability, i.e. the probability that the final probe measurement outcome differs from the true signal bit due to either a probe initialization error, an $X$-error before the phase gate, or a probe measurement error. More explicitly, 
\begin{equation}
    p_\oplus = \frac{1}{2}\left[1 - (1-2\wp_\text{prep})(1-2p_s)(1-2\wp_\text{meas}) \right]\,. 
\end{equation}
We want to find the FI for the probability distribution $\mathscr{P}({{x }};\theta)$ in \eqref{eq:logical-dist}. First, one can readily verify that 
\begin{multline}
  \left.  \frac{d}{d\theta}  \mathscr{P}({{x }};\theta)\right|_{\theta = 0} = \frac{{{x }}}{2}(1-2\wp_\text{meas})(1-2\wp_\text{prep}) \cdot \\ \sum_{\substack{w_0 w_1w_2}} P(w_0,w_1,w_2) \cdot m_{w_2} (1-2p_s) \,.
\end{multline}
Using 
\begin{equation}
    P(w_0,w_1,w_2) =  P(w_2|w_1) P(w_1|w_0) P(w_0)\,,
\end{equation}
where $P(w_i |w_{i-1})$ denotes the conditional probability of obtaining $w_i$ errors given the previous step had $w_{i-1}$ errors. We can write 
\begin{align*}
   & \sum_{\substack{w_0 w_1w_2}} P(w_0,w_1,w_2) \cdot m_{w_2} (1-2p_s) \\
   &= \frac{1}{n}\sum_{w_0} P(w_0)  \left[\sum_{w_1} (n-2w_1) P(w_1|w_0) \right] \\
   & \qquad \qquad \qquad \times \left[\sum_{w_2}(n-2w_2) P(w_2|w_1) \right]\,.
\end{align*}
The sum in the brackets can be calculated one by one
\begin{align*}
   & \sum_{w_2}(n-2w_2) P(w_2|w_1) \\
   &= (1-2p)(n-2w_1) \sum_{w_1} (n-2w_1)^2 P(w_1|w_0) \\
    &=  4np(1-p) + (1-2p)^2(n-2w_0)^2 \,,
\end{align*}
leading to 
\begin{align*}
     & \sum_{\substack{w_0 w_1w_2}} P(w_0,w_1,w_2) \cdot m_{w_2} (1-2p_s) \\
     &= \frac{(1-2p)}{n} \left[(1-2p)^2 \overline{m^2} + 4np(1-p) \right]\,. \numberthis
\end{align*}
The FI at $\theta = 0$ is given by 
\begin{align}
     &\mathcal{F}_\text{cl}^{(\rm QEC)} = \left[(1-2\wp_\text{meas})(1-2\wp_\text{prep})(1-2p) \right. \nonumber \\
     & \qquad \qquad \left. \left\{(1-2p)^2 \frac{\overline{m^2}}{n} + 4p(1-p)\right\} \right]^2\,, \label{eq:final-fi-0-exact} \\ 
    &= \left[(1-2\wp_\text{meas})(1-2\wp_\text{prep})(1-2p)^3\right]^2 \bigg(\frac{\overline{m^2}}{n}\bigg)^2 \nonumber \\
    & \qquad \qquad + O\bigg(\frac{\overline{m^2}}{n}\bigg)\,.
    \label{eq:final-fi-0}
\end{align}
For error rate below the state preparation threshold $p_\text{th}^{(s)}$, we can use \eqref{eq:rms approx av} to write 
\begin{align} 
       & \mathcal{F}_\text{cl}^{(\rm QEC)}  \\
       & \approx  \left[(1-2\wp_\text{meas})(1-2\wp_\text{prep})(1-2p) \right. \nonumber \\
       & \qquad  \left. \left\{((1-2p)(1-\alpha \mathsf{p}^2))^2n^2 + 4p(1-p)\right\} \right]^2\,, \label{eq:final-fi-exact} \\
        &\approx \left[(1-2\wp_\text{meas})(1-2\wp_\text{prep})(1-2p)^3\left(1-2\alpha \mathsf{p}^2\right)^2\right]^2 n^2\,. \label{eq:final-fi}
\end{align}

\eqref{eq:final-fi-0} and \eqref{eq:final-fi} demonstrate the attainability of the HL when the state preparation step is below error threshold (such that ${\overline{m^2}} = \Omega(n^2)$) and the measurement step is below the error threshold (such that $\wp_\text{meas}$ is below a constant). 

\begin{figure}
    \centering
    \includegraphics[width=0.90\linewidth]{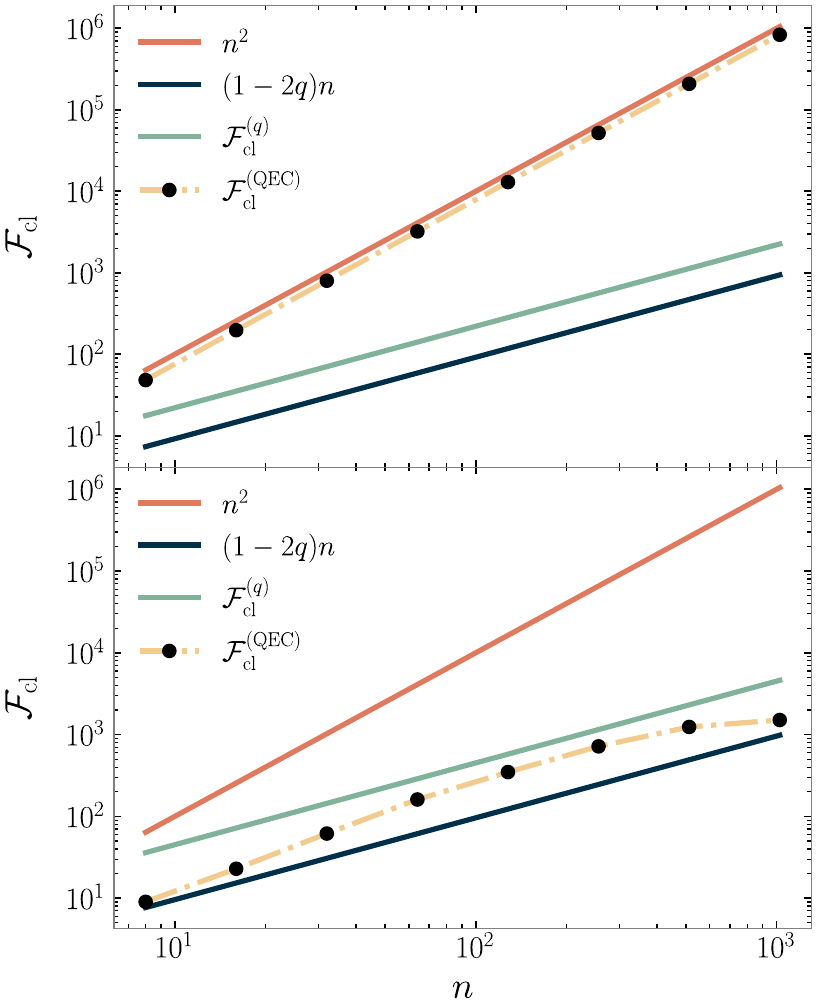}
    \caption{We plot the FI in different settings, including the noiseless case, the noisy case with our QEC protocol and the noisy case with $\ket{+}^{\otimes n}$ and $\ket{i\mathrm{GHZ}}^{\otimes n/k}$ as input states. The qubit and measurement error rates are chosen such that for \textbf{Top}: error rate is below both state preparation and measurement threshold (numerically, $\mathsf{p}=\mathsf{q}=0.04$ and $p=q=p_\text{prep} = 0.01$) and \textbf{Below}: error rate is below measurement but not state preparation threshold (numerically, $\mathsf{p}=\mathsf{q}=0.1$ and $p=q=p_\text{prep} = 0.02$). Below the threshold, our protocol $\mathcal{F}_\text{cl}^\text{(QEC)}$  (\eqref{eq:final-fi-0-exact}) asymptotically achieves the HL.}
    \label{fig:fisher-information} 
\end{figure}

We further compare our QEC protocol to other protocols in noisy settings without QEC. Firstly, consider the case where the probe is prepared in the product state $|+\rangle^{\otimes n}$. In noiseless setting, the associated FI scales as $n$, corresponding to SQL. If we consider the case of final probe state measurement error, then FI reduces to $(1-2q)n$. A more advanced protocol prepares $n/k$ copies of $k$-qubit $|i\text{GHZ}\rangle$ state which result in FI
\begin{equation}\label{eq:optimal FI measurement error}
   \mathcal{F}_\text{cl}^{(q)} = \frac{n}{k}\cdot (1-2q)^{2k} k^2,
\end{equation}
where $(1-2q)^{2k} k^2$ represents the FI of each $|i\text{GHZ}\rangle$ state. (Note that here we ignore the rounding error caused by the fact that $k$ and $n/k$ may not integers because our estimate of $\mathcal{F}_\text{cl}^{(q)}$ will only be larger than the one in practice.) The FI in \eqref{eq:optimal FI measurement error} can be optimized with respect to $k$, leading to 
\begin{equation}
    \mathcal{F}_\text{cl}^{(q)} = \frac{-n}{e\ln((1-2q)^2)} = \frac{n}{4qe} + O(1) \,,
\end{equation}
and occurs at $k= -1/\ln((1-2q)^2)$.

\vspace{0.1cm}

\noindent \textit{Numerics}.--- In \figref{fig:fisher-information}, we show the numerically calculated FI for our protocol. We consider an $n$-qubit probe state and error rates are chosen such that the state preparation and measurement errors are both below the threshold. In this regime, the FI asymptotically achieves  the HL with logarithmic overhead in circuit depth due to $O(\ln n)$ repeated measurements during state preparation and probe measurement. 

\section{Conclusion and Outlook}\label{sec:conclusion}

In this work, we show that quantum metrology admits a fully fault-tolerant regime. We use repetition code, and explicitly accounting for noise not only in the probes but also in the error-correction operations and measurements. We established the existence of a threshold error rates below which the HL can be asymptotically recovered. Our protocol achieves this with only logarithmic overhead in circuit depth, showing that high-precision quantum sensing remains possible even under realistic faulty operations.


The signal Hamiltonian (Pauli-Z) considered in this work, is relevant to practical cases such as optimal interferometry or spin ensembles and our results demonstrate the fault-tolerance of Pauli-Z estimation against bit-flip noise. When dephasing noise is present, however, no-go results forbid the achievability of the HL. One future direction is to consider the non-asymptotic advantage of QEC in metrology in a noise-intermediate regime where the qubit number is limited and a moderate-level dephasing noise exists. Another future direction is to consider other other sensing scenarios, e.g., many-body Hamiltonian estimation where single-qubit errors might be tolerable and precision beyond the HL is achievable, referred to as super-Heisenberg limit~\cite{PhysRevLett.100.220501,PhysRevLett.98.090401,Napolitano2011}.

\begin{acknowledgments}

S.Z. would like to thank Sophie Kadan for helpful discussions.
H.S. and S.Z. acknowledge support from National Research Council of Canada (Grant No.~AQC-217-1) and Perimeter Institute for Theoretical Physics, a research institute supported in part by the Government of Canada through the Department of Innovation, Science and Economic Development Canada and by the Province of Ontario through the Ministry of Colleges and Universities. 
Q.X. is funded in part by the Walter Burke Institute for Theoretical Physics at Caltech.
\end{acknowledgments}

\vspace{0.1cm}
\noindent {\textbf{\textsf{Data Availability}}}
\vspace{0.1cm}

\noindent The data that support the findings of this paper are openly available.

\appendix
\section{Summary of error parameters and notation}\label{appendix:notation}

\renewcommand{\arraystretch}{1.25}
\begin{table}
    \centering
    \caption{\label{table:notations}Summary of error parameters and notation used throughout the paper.}
    \rowcolors{-1}{}{gray!10}
    \begin{tabularx}{\columnwidth}{c X}
        \toprule
        \textbf{Symbol} & \textbf{Meaning} \\
        \midrule
        $p$ & Physical Pauli-$X$ error per circuit location \\
        $q$ & Single-qubit measurement outcome error \\
        $p_{\mathrm{prep}}$ & Single-qubit initialization error \\
        $\mathsf{p}$ & Effective data-qubit Pauli-$X$ error accumulated during a syndrome measurement round \\
        $\mathsf{q}$ & Effective syndrome measurement error probability \\
        $q_x$ & Error probability of an ancilla-assisted, non-destructive Pauli-$X$ measurement \\
        $\wp_{\mathrm{prep}}$ & Residual logical preparation error of the probe state prepared from 
        \(|+\rangle^{\otimes n}\) via repeated ancilla-assisted non-destructive measurements \\
        $\wp_{\mathrm{meas}}$ & Logical measurement error probability for the logical operator 
        \(X^{\mathsf{R}}_{\mathsf{L}}\) acting on the probe state \\
        $p_s^{(\text{th})}$ & Threshold effective bit-flip error rate $\mathsf{p}$, below which the normalized rms magnetization satisfies \eqref{eq:rms approx av} \\
        $q_m^{(\text{th})}$ & Threshold Pauli-$X$ measurement error rate $q_x$, below which the logical measurement error $\wp_{\mathrm{meas}}$ can be made arbitrarily small by repeated non-destructive measurements \\
        \bottomrule
    \end{tabularx}
\end{table}

In this section, we summarize the notation and error parameters used throughout the paper. In Table\,\ref{table:notations} lists all physical and effective error probabilities appearing in the main text and appendices, together with their physical interpretations. This summary is provided for ease of reference and to clarify the distinction between circuit-level error rates and the corresponding phenomenological parameters used in our analysis.

\section{Syndrome measurement error mechanism}\label{appendix:syndrome measurement}

In this section, we analyze the error mechanism arising in the syndrome measurement procedure. In particular, we identify how physical Pauli-$X$ faults and SPAM error at the circuit level give rise to effective syndrome and bit-flip errors, and justify the phenomenological syndrome error model used in the main text. 

\begin{figure}
    \centering
    \includegraphics[width=0.75\linewidth]{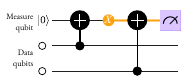}
    \caption{Illustrate of a single-fault error propagation in the syndrome measurement circuit. Orange lines highlight how a Pauli-$X$ error propagates through the circuit and leads to an incorrect syndrome measurement by inducing an $X$-error on the measure qubit before readout.}
    \label{fig:error-propagation}
\end{figure}

At first, we note that no $Z$-error are present on the measure qubit, and any $Z$-error may present on a data qubit does not propagate to the measure qubit. Similarly, an $X$-error on the measure qubit does not propagate to the data qubits. As a result, the syndrome measurement circuit does not generate hook errors.

Effective syndrome error $\mathsf{q}$ arise from error events that lead to an incorrect measurement outcome, equivalently from processes that induce an $X$-error on the measurement qubit immediately prior to readout or a readout error (but not both). In \figref{fig:error-propagation}, we illustrate one such mechanism that flips the syndrome measurement value. Summing all physical error contributions---namely Pauli-$X$ errors and SPAM errors on the measure qubit---we obtain
\begin{equation}
\mathsf{q} = p_\text{prep} + 3p + q + \text{h.o} \,,  
\end{equation}
where $+\text{h.o.}$ means terms of hight than linear order in the physical error probabilities.

To calculate effective qubit error $\mathsf{p}$, we sum all the physical error leading to syndrome detection, leading to
\begin{equation}
\mathsf{p} =  3p + \text{h.o.} \,.    
\end{equation}
Lastly, we note that in the syndrome-extraction circuit, error occurring on the measure qubit do not propagate to the data qubits. As a result, error mechanism leading to incorrect syndrome outcomes are independent from those inducing Pauli-$X$ error on the data qubits.

\begin{lemma}[Independence of data and syndrome errors]
In the syndrome-extraction circuit under the noise model defined in Def.\,\ref{def:noise model}, the error mechanisms giving rise to syndrome errors and to bit-flip errors on the data qubits are independent.    
\end{lemma}
\begin{proof}
    The measure qubit interacts with data qubits only through CNOT that do not propagate Pauli-$X$ errors from the measure to the data qubit. Consequently, physical Pauli-$X$ and SPAM error on the measure qubit only modify the measure syndrome value and do not affect the data qubits.
\end{proof}

\section{Probe initialization error}\label{appendix:probe initialization error}

In probe state preparation step, we begin with preparing the product state $|+\rangle^{\otimes n}$. However, as a result of state preparation error, a state is wrongly prepared in $|-\rangle$ state with probability $p_\text{prep}$. These state preparation error can be thought as $Z$-error, therefore, an extensive amount of these would drastically decay the scaling of FI. 

\begin{figure}
    \centering
    \includegraphics[width=0.75\linewidth]{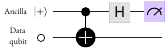}
    \caption{Quantum circuit for measuring Pauli-$X$ operator on data qubit using an ancilla. }
    \label{fig:measure-X}
\end{figure}

To increase the state preparation fidelity, we repeatedly measure Pauli-$X$ operator on each probe qubit using the quantum circuit in \figref{fig:measure-X}. The error associated with single Pauli-$X$ measurement is given by: 
\begin{equation}
    q_{x} = p_\text{prep}(1-q) + (1-p_\text{prep}) q \approx p_\text{prep} + q\,.
\end{equation}
Crucially, such a circuit does not propagate any $Z$-error (associated with ancilla state preparation) into probe qubit, while any $X$-error does not affect the probe qubit. Let us suppose we perform $r$ such measurements and perform the majority voting to determine the state. Then, the failure probability given by: 
\begin{align*}
    \sum_{k= \lceil r/2\rceil}^r \binom{r}{k} q_x^k(1-q_x)^{r-k} &\leq  \sum_{k= \lceil r/2\rceil}^r \binom{r}{\lceil r/2\rceil}  q_x^k \,, \\
     &= 2^r q_x^{\lceil r/2\rceil}\,, \\
     &\leq 2^r q_x^{r/2+1}. 
\end{align*}
Then, the probability of failure $\wp_\text{prep}$ in computing the complete set of $X$-measurements given by 
\begin{equation}
    \wp_\text{prep} \leq 1 - \left( 1- 2^r q_x^{r/2+1} \right)^n\,.
\end{equation}
When $q_x<1/4$, we can pick e.g. $r = -\frac{\ln({q_x^2n^4})}{\ln(4 q_x)}$, then 
\begin{equation}
\label{eq:step1 measurement}
\wp_\text{prep} \leq 1 - \left(1 - 1/ n^2\right)^{n},    
\end{equation}
which approaches $0$ for large $n$.

\begin{figure}
    \centering
    \includegraphics[width=0.49\linewidth]{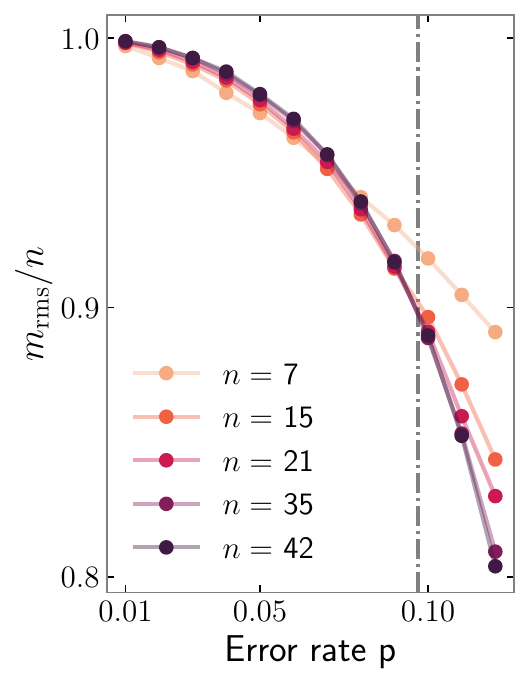}
    \includegraphics[width=0.49\linewidth]{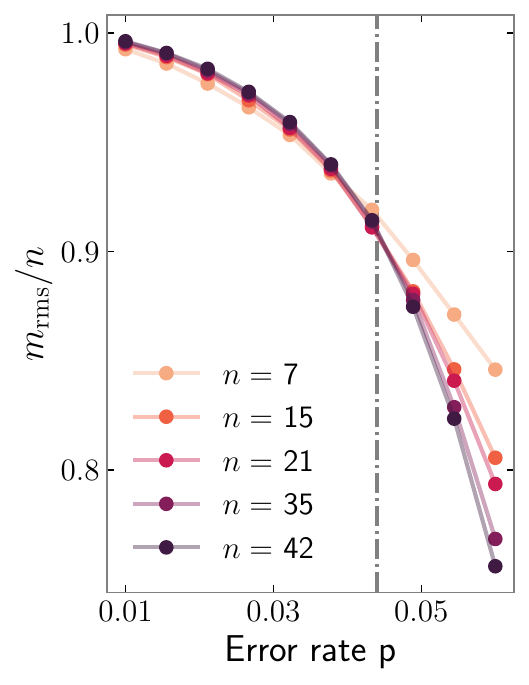}
    \caption{The ratio $m_\text{rms}/n$ as a function of qubit error rate $\mathsf{p}$ for (a) $\mathsf{q}=\mathsf{p}/2$ (b) $\mathsf{q}=2\mathsf{p}$ for varying system-sizes $n$. The vertical dashed-dot line shows the threshold probability $p^{(s)}_\text{th}$ with calculated values $0.092$ and 0.051, respectively.}
    \label{fig:state-prep-p-not-q}
\end{figure}

\section{State preparation when \texorpdfstring{$\mathsf{p} \not =\mathsf{q}$}{TEXT}} \label{appendix:state prep p not q}

We consider the state preparation in case where the qubit and measurement error rate are not equal to each other and show the existence of threshold for ratio $m_\text{rms}/n$.

We do the similar qualitative analysis to understand when does the matching algorithm fails. Consider the same case as in discussed in the main text shown in \figaref{fig:decoding_state_prep} where two consecutive defects are produced. This could result from either a single-qubit error which occurs with probability $\mathsf{p}$ or two consecutive measurement errors which occurs with probability $\mathsf{q}^2$. The decoder chooses 1-chain with probability $\max\,\{\mathsf{p},\mathsf{q}^2\}$ so that it fails with probability $\mathsf{p}_\star \equiv \min\,\{\mathsf{p},\mathsf{q}^2\}$. The error mechanism leads to single qubit bit-flip error that scale with system-size remain same, show in \,\figbref{fig:decoding_state_prep}. At small error rate $p$, the rms magnetization can be written analogous to  \eqref{eq:rms approx av} as
\begin{equation}
    \frac{m_\text{rms}}{n} \approx  1- 2\alpha \mathsf{p}_\star \,,
\end{equation}
for some constant $\alpha$. In \,\figref{fig:state-prep-p-not-q}, we show the calculated ratio for (a) $\mathsf{q}=\mathsf{p}/2$ and (b) $\mathsf{q}=2\mathsf{p}$, with qubit error $\mathsf{p}$ for varying system-sizes $n$. In both cases, we find a threshold probability $p^{(s)}_\text{th}$ below which the ratio obeys $(1-\gamma \mathsf{p}^2)$ (not explicitly shown in figure) for some constant $\gamma$. 

\begin{figure}
    \centering
    \includegraphics[width=0.9\linewidth]{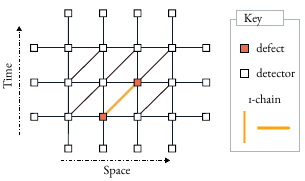}
    \caption{Detector error graph for the noise model defined in Def.\,\ref{def:correlated error model}, shown for a distance-5 repetition qubit with four rounds of syndrome measurements. Vertical and horizontal edges corresponds to single-location faults, while diagonal edges represents correlated post-CNOT errors (i.e., $X\otimes X$), which induce correlated detection events across space and time.}
    \label{fig:correlated-dem}
\end{figure}

\begin{figure}
    \centering
    \includegraphics[width=0.8\linewidth]{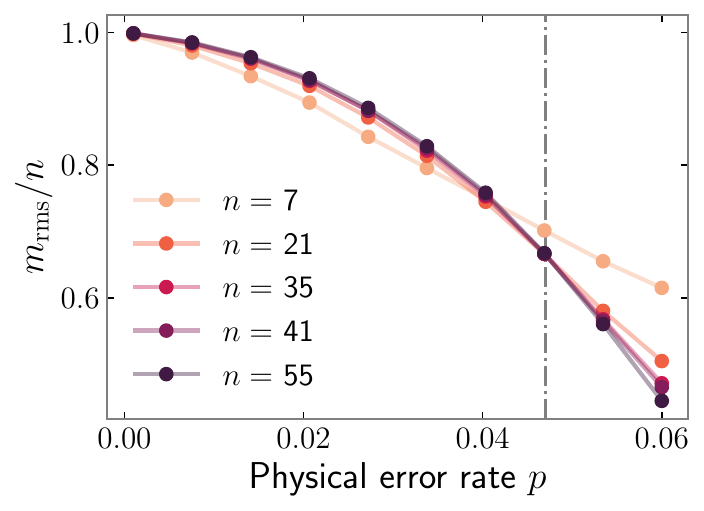}
    \caption{The ratio $m_\text{rms}/n$ as a function of physical error rate $p=q=p_\text{prep} = 5p_\text{CNOT}$ calculated for correlated circuit-level noise model in Def.\,\ref{def:correlated error model}, where $\lceil 5 \ln n\rceil$ rounds of repeated syndrome measurements are performed. The vertical dashed-dot line shows the threshold probability $p_\text{th}^{(s)}$ with calculated value $0.47$ obtained by calculating crossover at finite system-size.}
    \label{fig:msp-correlated}
\end{figure}
\section{State preparation under correlated errors}\label{appendix: correlated error}

In this section, we consider the state preparation under a more generic bit-flip noise model defined as:
\begin{definition}\label{def:correlated error model}
We consider a circuit-level noise model in which faults may occur at any physical location, including state preparation, gate, idle, and measurement operations. The noise processes are defined as follows:
\begin{itemize}
\item \textbf{State preparation.} A physical qubit initialized in the state $|0\rangle$ (or $|+\rangle$) is instead prepared in $|1\rangle$ (or $|-\rangle$) with probability $p_\text{prep}$.
\item \textbf{Measurement.} Each single-qubit measurement is a projective measurement in the computational ($Z$) basis, with the reported eigenvalue flipped with probability $q$.
\item \textbf{Gate and idle errors.} Each physical single-qubit gate and idle (wait) operation is followed by a Pauli-$X$ error with probability $p$. After each CNOT gate, an error drawn from the set $\{I\otimes X, X\otimes I, X\otimes X\}$ occurs with probability $p_\text{CNOT}$.
\end{itemize}
All error events are assumed to occur independently unless otherwise stated.
\end{definition}
A correlated post-CNOT error (i.e., an $X\otimes X$ error) results in correlated syndrome and Pauli-$X$ error on data qubit. In the detector graph, such a error corresponds to an edge connecting diagonally opposite detectors at time $t$ and $t+1$ as shown in \figref{fig:correlated-dem}. To leading order, any correlated errors that are not correctly identified by the decoder in the final round of syndrome measurements leads to effective single-qubit bit-flip errors whose total contribution scales extensively with the system size. As a result, for sufficiently small physical error rates, there exists a constant $\epsilon$ such that
\begin{equation}
    \frac{m_\text{rms}}{n} \approx 1-\epsilon\,,
\end{equation}
where $\epsilon$ is larger than in the case without correlated errors, reflecting the increased rate of effective single-qubit errors.

Furthermore, the presence of such short-range correlated errors does not modify the qualitative structure of the error distribution or introduce long-range correlations. Therefore, our analysis of nonlocal error suppression continues to hold.

In \figref{fig:msp-correlated}, we show the calculated ratio $m_\text{rms}/n$ for the circuit-level noise model with $p_\text{prep} = p = q = 5p_\text{CNOT}$ for varying system-sizes $n$. We find the threshold probability $p_\text{th}^{(s)}$ below which the ratio approaches a constant value. As a consequence, the HL scaling of the FI is preserved even in the presence of these correlated circuit-level errors.


\bibliographystyle{apsrev4-1}
\bibliography{biblio} 

\end{document}